\documentclass[letterpaper, 10 pt, conference]{ieeeconf}  
\IEEEoverridecommandlockouts                              
\usepackage{graphicx}
\usepackage{amsmath}
\usepackage{amsfonts}
\usepackage[acronym]{glossaries}
\usepackage{xcolor}
\usepackage{url}
\usepackage{cite}

\def\wf{w_\mathrm{f}}
\def\wi{w_\mathrm{ini}}
\def\wir{\mathrm{w}_\mathrm{ini}}
\def\yi{y_\mathrm{ini}}
\def\ui{u_\mathrm{ini}}
\def\uf{u_\mathrm{f}}
\def\ufr{\mathrm{u}_\mathrm{f}}
\def\yf{y_\mathrm{f}}
\def\wr{w_\mathrm{ref}}
\def\ur{u_\mathrm{ref}}
\def\yr{y_\mathrm{ref}}

\def\Wd{W}
\def\mp{\mu_\mathrm{pred}}
\def\mph{\hat{\mu}_\mathrm{pred}}
\def\mw{\mu^\mathrm{w}_t}
\def\mdep{\mu^\mathrm{dep}_t}
\def\mfree{\mu^\mathrm{free}_t}
\def\wfree{w^\mathrm{free}_t}
\def\wdep{w^\mathrm{dep}_t}
\def\Sw{\Sigma^\mathrm{w}}
\def\Swh{\hat{\Sigma}^\mathrm{w}}
\def\Sp{\Sigma_\mathrm{pred}}
\def\Sph{\hat{\Sigma}_\mathrm{pred}}
\def\dy{d_\mathrm{y}}

\def\uft{\tilde{u}_\mathrm{f}}

\newtheorem{thm}{Theorem}

\newtheorem{prop}{Proposition}
\newtheorem{definition}{Definition}
\newtheorem{lemma}{Lemma}
\newtheorem{rmk}{Remark}

\newacronym{LTI}{LTI}{linear time-invariant}
\newacronym{SPC}{SPC}{subspace predictive control}
\newacronym{DeePC}{DeePC}{data-enabled predictive control}
\newacronym{KL}{KL}{Kullback–Leibler}
\newacronym{MPC}{MPC}{model predictive control}

\def\arxiv{1} % 0 if short version, 1 if arXvi version 

\title{\LARGE \bf
Gaussian behaviors: representations and data-driven control
}

\author{Andr\'as Sasfi, Ivan Markovsky, Alberto Padoan, Florian D\"orfler% 
\thanks{A. Sasfi and F. D\"orfler are with the Department of Information Technology and
Electrical Engineering, ETH Z\"urich, 8092 Z\"urich, Switzerland (e-mail: \{asasfi, doerfler\}@control.ee.ethz.ch). I. Markovsky is with the Catalan Institution for Research and Advanced Studies, 08010 Barcelona, Spain, and also with the International Centre for Numerical Methods in Engineering, 08034 Barcelona, Spain (e-mail: ivan.markovsky@cimne.upc.edu). A. Padoan is with the Department of Electrical \& Computer Engineering, University of British Columbia. (e-mail: apadoan@ece.ubc.ca). This work was supported by the SNF/FW Weave Project 200021E\_20397}
}

\begin{document}

\maketitle
\thispagestyle{empty}
\pagestyle{empty}

\begin{abstract}
We propose a modeling framework for stochastic systems, termed Gaussian behaviors, 
that describes finite-length trajectories of a system as a Gaussian process.
The proposed model naturally quantifies the uncertainty in the trajectories, yet it is simple enough to allow for tractable formulations.
We relate the proposed model to existing descriptions of dynamical systems including deterministic and stochastic behaviors, and \acrlong{LTI} state-space models with Gaussian noise.
Gaussian behaviors can be estimated directly from observed data as the empirical sample covariance.
The distribution of future outputs conditioned on inputs and past outputs provides a predictive model that can be incorporated in predictive control frameworks.
We show that \acrlong{SPC} is a certainty-equivalence control formulation with the estimated Gaussian behavior.
Furthermore, the regularized \gls{DeePC} method is shown to be a {\textit{distributionally optimistic}} formulation that optimistically accounts for uncertainty in the Gaussian behavior.
To mitigate the excessive optimism of \gls{DeePC}, we propose a novel distributionally robust control formulation, and provide a convex reformulation allowing for efficient implementation.
\end{abstract}

\section{Introduction}
Recent data-driven control methods based on the behavioral approach~\cite{coulson2019data,coulson2021distributionally,berberich2020data,breschi2023data,depersis2019formulas} have gained significant attention.
These formulations rely on behavioral systems theory that treats systems as sets of trajectories, and allows to represent \acrfull{LTI} systems directly with data~\cite{willems1997introduction,willems2005note}.
The methods exploit this data representation and typically add regularization to the problem~\cite{dorfler2022bridging,markovsky2022data} for a posteriori robustification. 
The resulting formulations can achieve comparable or even superior performance compared to classical methods consisting of an identification and a control step, especially in challenging scenarios involving nonlinear effects and large amounts of uncertainty~\cite{dorfler2022bridging,markovsky2021behavioral}.
Existing works explain this observation by the robustifing effect of regularization, yet they rely on the inherently deterministic behavioral description~\cite{huang2021quadratic,coulson2021distributionally,berberich2020data}.
On the other hand, only few works~\cite{willems2012open,baggio2017lti,faulwasser2023behavioral,pan2022stochastic,fiedler2023probabilistic,chiuso2025harnessing} propose frameworks that start directly from a stochastic system model.
In fact, the survey~\cite{markovsky2021behavioral} points out that the extension of data representations to the stochastic domain is still the subject of fervent debate.

Stochastic extensions to behavioral systems theory have been defined bottom up in the literature~\cite{willems2012open,baggio2017lti}.
However, the existing works provide general and abstract (and thus also blunt) perspectives that hinder the practical applicability of the frameworks.
Stochastic behaviors have also been modeled in the literature using polynomial chaos expansions~\cite{faulwasser2023behavioral,pan2022stochastic}.
However, the complexity of the resulting methods grow with the order of the expansion, limiting scalability.
Recently, a stochastic interpretation of data-driven control formulations was proposed in~\cite{chiuso2025harnessing}, highlighting that regularization accounts for the uncertainty in a linear model estimated from data.
In this work, we take a different approach, and propose a new bottom-up stochastic behavioral modeling framework that admits a data representation and leads to tractable control formulations which can be solved efficiently.

The contributions of this work are the following.
First, we propose a stochastic modeling framework, termed \textit{Gaussian behaviors}, that describes finite-length trajectories as normally distributed random vectors.
Gaussian behaviors can be estimated directly as the empirical sample covariance, and they provide predictive models in the form of conditional distributions.
For \gls{LTI} systems, the deterministic behavioral description and a stochastic state-space representation with Gaussian initial state, input, and noise are special cases of the proposed definition by imposing additional structure on the covariance.
Second, we propose data-driven predictive control formulations using Gaussian behaviors and explore connections to existing frameworks.
In particular, \acrfull{SPC}~\cite{favoreel1999spc} is a certainty-equivalence formulation in this framework.
To account for inaccuracies in the predictive model, we optimize the predictive distribution within an uncertainty set centered around the estimate.
First, we propose a \textit{distributionally optimistic} formulation, which is shown to be equivalent to the regularized \acrfull{DeePC} approach.
Next, we propose a distributionally robust formulation leading to a min-max problem that can be reformulated as a convex problem, enabling efficient computation.

\section{Preliminaries}
\subsection{Notation}
The image and the Moore–Penrose inverse of a matrix $A$ are denoted by $\mathrm{im}(A)$ and $A^\dagger$, respectively.
The set of symmetric positive (semi) definite matrices is denoted by $\mathbb{S}_{>0}~(\mathbb{S}_{\geq0})$.
The identity matrix of dimension $n \times n$ is $I_{n}$.
Furthermore, the (semi-)norm of $x$ weighted by a matrix $M\in\mathbb{S}_{>0}~(\mathbb{S}_{\geq0})$ is denoted by $\|x\|_{M}= x^\top M x$. A realization of the random variable $x$ is denoted by $\mathrm{x}$.

\subsection{Behavioral systems theory for deterministic \gls{LTI} systems} \label{sec:BST}
Behavioral systems theory abstractly defines systems as a set of trajectories, called the \textit{behavior}, regardless of its representation.
We focus on system trajectories of a given finite length $L>0$, which we denote by $w_{[t,t+L-1]}=[w_t^\top~w_{t+1}^\top~\dots~w_{t+L-1}^\top]^\top$, or simply by $w$ in this section.
For \gls{LTI} systems, the \textit{restricted behavior}, $\mathcal{B}_L$, is the set of all system trajectories of length $L$ defined as
\begin{align*}
    \mathcal{B}_L = \{w\in \mathbb{R}^{qL}~|~w \text{ is a length-$L$ trajectory of the system}\},
\end{align*}
which is a shift-invariant subspace of the set of all possible trajectories~\cite{markovsky2022data}.
For \gls{LTI} systems and for large enough $L$, $\mathcal{B}_L \subseteq \mathbb{R}^{qL}$ is a subspace of dimension $d=mL + n$, where $n$ is the order of a minimal system realization and $m$ is the number of inputs~\cite{willems1997introduction}.

A basis for the subspace $\mathcal{B}_L$ can be readily constructed from a state-space representation for \gls{LTI} systems~\cite{markovsky2022data}.
Furthermore, sufficiently rich data directly provide representations of deterministic \gls{LTI} behaviors.
Assume that measurements of $D>qL$ trajectories $w^i\in\mathbb{R}^{qL},~i\in\{1,\dots,D\}$, from the system are available, and define the data matrix $\Wd = [w^1~\dots~w^D]$. 
According to~\cite{van2020willems}, if the inputs are \textit{collectively persistently exciting}, $\Wd$ has rank $mL+n$, and it spans the restricted behavior, i.e., $   \mathcal{B}_L = \mathrm{im}(\Wd)$.

\subsection{Stochastic \gls{LTI} systems}
Consider a stochastic \gls{LTI} state-space model
\begin{align} \label{eq:stochastic-SS}
\begin{split}
     x_{t+1} &= Ax_t + Bu_t + \xi_t,\\
    y_t & = Cx_t + D u_t + \eta_t,
\end{split}
\end{align}
where $\xi_t\in\mathbb{R}^n$ and $\eta_t\in\mathbb{R}^p$ are the process and measurement noise, respectively.
Assume that both $\xi_t$ and $\eta_t$ are independent and identically distributed with zero mean and covariance $\Sigma^\xi$ and $\Sigma^\eta$, respectively.
The output trajectory $y=[y_t^\top~y_{t+1}^\top~\dots~y_{t+L-1}^\top]^\top$ is a random variable and can be expressed as a linear combination of $x_t$ and the input and noise trajectories $u,~\xi,~\eta$, defined similarly to $y$, leading to 
\begin{align} \label{eq:stochastic_ss_linear}
    y = \mathcal{O}_L x_t + \mathcal{T}_L^\mathrm{u} u + \mathcal{T}_L^\xi \xi + \eta,
\end{align}
where 
\small
\begin{align*}
    \mathcal{O}_L = \begin{bmatrix}
        C \\ CA \\ \vdots \\ CA^{L-1}
    \end{bmatrix},~ \mathcal{T}_L^\mathrm{u} = 
    \begin{bmatrix}
      D & 0 & 0 & \dots & 0 \\
        CB & D & 0 & \dots & 0 \\
        CAB & CB & D &\dots  & 0\\
        \vdots & \vdots & \ddots & \ddots & \vdots \\
        CA^{L-2}B & \dots & \dots & CB & D   
  \end{bmatrix},
\end{align*}
\normalsize
and $\mathcal{T}_L^\xi$ is defined as $\mathcal{T}_L^\mathrm{u}$ with $B=I,~D = 0$.

\subsection{Data-driven predictive control formulations} \label{sec:prelim-dd-control}
Various data-driven predictive control formulations exploit the data representation $\mathcal{B}_L\in\mathrm{im}(W)$ of the behavior~\cite{coulson2019data,favoreel1999spc,breschi2023data,depersis2019formulas}.
Without loss of generality, each trajectory $w_t\in\mathbb{R}^q$ of an \gls{LTI} system can be partitioned into inputs $u_t\in\mathbb{R}^m$ and outputs $y_t\in\mathbb{R}^{q-m}$.
To ensure compatibility with the most recent measurements, the initial part of the trajectory $\wi = [\ui^\top~\yi^\top]^\top$ of length $(L - L_\mathrm{f})$ is considered to be given, and the future part $\wf = [\uf^\top~\yf^\top]^\top$ of length $L_\mathrm{f}$ is a free variable to be optimized.
With a suitable permutation of the rows, the whole trajectory is $w = [\wi^\top~\wf^\top]^\top$.
The control objective is often posed as minimizing the cost
\begin{align*}
    J(\wf) = (\wf - \wr)^\top \begin{bmatrix}R& 0 \\ 0 & Q 
    \end{bmatrix}(\wf-\wr),
\end{align*}
where $\wr\in\mathbb{R}^{qL_\mathrm{f}}$ is a reference trajectory, and $R\in\mathbb{S}_{>0}$, $Q\in\mathbb{S}_{\geq0}$ are weight matrices for the entire future input and output trajectories, respectively.
The future trajectory $\wf=[\uf^\top~\yf^\top]^\top$ is often restricted to a non-empty, closed and convex constraint set $\mathcal{W} = \mathcal{U} \times \mathcal{Y}$.

We review two data-driven predictive control formulations below. 
Let us partition the rows of the data matrix $\Wd$ as $\Wd = [W_\mathrm{p}^\top,U_\mathrm{f}^\top,Y_\mathrm{f}^\top]^\top$, where the block rows correspond to $\wi,\uf$ and $\yf$, respectively.
The classical \gls{SPC} formulation~\cite{favoreel1999spc,markovsky2022data} is then given as
\begin{align} \label{eq:SPC}
\begin{split}
    \min_{\wf\in \mathcal{W}} & \quad  J(\wf) \\
    \mathrm{s.t.} & \quad \yf = Y_\mathrm{f} \begin{bmatrix}
        W_\mathrm{p} \\ U_\mathrm{f}
    \end{bmatrix}^{\dagger} \begin{bmatrix}
        \wi \\ \uf
    \end{bmatrix}.
\end{split}
\end{align}
\vspace{2mm}
Furthermore, the \gls{DeePC} method from~\cite{dorfler2022bridging} is based on
\begin{align} \label{eq:DeePC}
\begin{split}
    \min_{\wf \in \mathcal{W},g\in\mathbb{R}^{D}} & \quad J(\wf) + \lambda_g \cdot h(g) \\
    \mathrm{s.t.} & \quad \begin{bmatrix}
        \wi \\ \uf \\ \yf
    \end{bmatrix} = \begin{bmatrix}
        W_\mathrm{p} \\ U_\mathrm{f}\\ Y_\mathrm{f}
    \end{bmatrix} g,
\end{split}
\end{align}
where $h(g)$ is a regularizer with weight $\lambda_g \geq 0$. 
Typical choices for the function $h(g)$ are the (squared) 1-norm, 2-norm, or the projected 2-norm given as $h(g) = \|(I-\Pi)g\|_2^2$, with $\Pi:=
([W_\mathrm{p}^\top~ U_\mathrm{f}^\top]^\top)^\dagger[W_\mathrm{p}^\top~ U_\mathrm{f}^\top]^\top$, see~\cite{markovsky2022data,dorfler2022bridging}.

\section{Gaussian behaviors} \label{sec:Gaussian_behaviors}
We propose a pragmatic yet foundational modeling framework for (discrete-time) stochastic dynamical systems based on stochastic processes that will be useful specifically in the data-driven setting.
Similarly to behavioral systems theory, we consider a trajectory-based definition.
In particular, we focus on finite-length trajectories and model them as Gaussian processes~\cite{doob1953stochastic,williams2006gaussian}, ensuring that the proposed model is tractable.
We consider each trajectory of length $L>0$ to be a random vector containing $L$ consecutive time instants of the random process $(w_t)_{t\in\mathbb{Z}}$ denoted by
\begin{align*}
    w_{[t,t+L-1]} = [w_{t}^\top~w_{t+1}^\top~\dots~w_{t+L-1}^\top]^\top.
\end{align*}

\begin{definition} \label{def:cov_model}
A \textit{(finite-length) Gaussian behavior} is a collection of random length-$L$ trajectories $w_{[t,t+L-1]}\in\mathbb{R}^{qL}$ distributed as a multivariate Gaussian process, that is,
\begin{align*}
    w_{[t,t+L-1]} \sim \mathcal{N}(\mw,\Sw_t), 
\end{align*}
with mean $\mw\in\mathbb{R}^{qL}$ and covariance $\Sw_t\in\mathbb{S}_{\geq0}$.
Furthermore, a Gaussian behavior is \textit{linear}, if 
\begin{align}
\mw \in \mathrm{im}(B_t), \label{eq:linear}
\end{align}
where $B_t\in\mathbb{R}^{qL\times d}$ is full rank and $d<qL$.
\end{definition}

The covariance $\Sw_t$ captures the correlation between the components of $w_{[t,t+L-1]}$.
This model naturally quantifies uncertainty in the trajectories and enables tractable methods to perform downstream tasks, such as prediction or control.
Furthermore, the empirical sample covariance can serve as an estimate for $\Sw_t$ (see Lemma~\ref{lemma:MLE}), leading to a data representation of the system, which was recently used~also in~\cite{zhao2024data}.
Lastly, this model is consistent with existing modeling frameworks for stochastic \gls{LTI} systems, as shown below.

\subsection{Prediction using Gaussian behaviors} \label{sec:pred}
In many downstream tasks, some parts of the trajectory is already observed, or fixed a priori.
For example, in predictive control, the initial part $\wir$ of the trajectory is observed, and the future inputs $\ufr$ are decision variables that we fix.
Let us partition $w_{[t,t+L-1]}$ into ``free" and ``dependent" variables, $\wfree$ and $\wdep$.
The probability of $\wdep$ \textit{given} $\wfree$ is expressed by the conditional distribution and can be easily computed for Gaussians.
After a suitable permutation of the rows, we partition $\mw$ and $\Sw_t$ into blocks corresponding to $\wfree$ and $\wdep$ as
\begin{align*}
    \mw = \begin{bmatrix}
    \mfree \\ \mdep
    \end{bmatrix}, \qquad 
    \Sw_t = \begin{bmatrix}
        \Sigma^\mathrm{ff}_t & (\Sigma^\mathrm{df}_t)^\top \\
        \Sigma^\mathrm{df}_t & \Sigma^\mathrm{dd}_t
    \end{bmatrix}.
\end{align*}
Then, the conditional distribution is  $\wdep|\wfree=\mathrm{w}^\mathrm{free}_t$ $\sim \mathcal{N}(\mp,\Sp)$ with mean and variance
\begin{align}
\mp & = \mdep + \Sigma^\mathrm{df}_t(\Sigma^\mathrm{ff}_t)^{\dagger} \left(\mathrm{w}^\mathrm{free}_t - \mfree\right), \label{eq:pred_mean_theory} \\
\Sp & = \Sigma^\mathrm{dd}_t - \Sigma^\mathrm{df}_t(\Sigma^\mathrm{ff}_t)^{\dagger}(\Sigma^\mathrm{df}_t)^\top.
\label{eq:pred_var_theory}
\end{align}
We call $\mathcal{N}(\mp,\Sp)$ the \textit{predictive distribution}.
Note that in case the free variables are deterministic, $\Sigma^\mathrm{ff}_t=0$. In this case, $\mp = \mdep$, and  $\mdep$ is related to $ \mfree$ through~\eqref{eq:linear} for linear Gaussian behaviors.

\subsection{Relation to existing system descriptions}
\subsubsection{Stochastic behaviors}
The notion of stochastic behaviors in~\cite{baggio2017lti} incorporates a probability space structure into the set of admissible system trajectories, and it includes as special cases both the deterministic notion of dynamical system and traditional stochastic processes~\cite{doob1953stochastic}. 
In this context, a \emph{stochastic dynamical system} is defined as a quadruple  
$
(\mathbb{T}, \mathbb{W}, \mathcal{F}, P),
$
where \(\mathbb{T}\) is the \emph{time axis} (e.g., \(\mathbb{T} = \mathbb{Z}\), \(\mathbb{T} = \mathbb{R}\)), \(\mathbb{W}\) is the \emph{signal set} (e.g.,  \(\mathbb{W} = \mathbb{R}^q\) or \(\mathbb{W} = \mathbb{C}^q\)), \(\mathcal{F}\) is a \emph{$\sigma$-algebra} of events, that is, a collection of measurable subsets of the space of trajectories \(\mathbb{W}^\mathbb{T}\) that satisfies certain properties, and \(P\) is a \emph{probability measure} on \((\mathbb{W}^\mathbb{T}, \mathcal{F})\), assigning probabilities to events, i.e., to entire sets of trajectories.
Opposed to stochastic processes, the behavioral approach~\cite{willems2012open,baggio2017lti} shifts the focus from individual random variables to the structure of admissible trajectories in \( \mathbb{W}^\mathbb{T} \) emphasizing \emph{event spaces}.

Our approach of modeling stochastic systems as Gaussian processes fits in this framework as a special case with finite time horizon $\mathbb{T} = \mathbb{Z}$ and $\mathbb{W} = \mathbb{R}^q$.
Furthermore, the $\sigma$-algebra of events \(\mathcal{F}\) is the Borel $\sigma$-algebra on $(\mathbb{R}^q)^{\mathbb{Z}}$, generated by the open sets under the product topology (i.e., the topology of pointwise convergence).
The probability measure \(P\) is that of a Gaussian process in the proposed framework.
Even tough we lose some of the generality of the definitions given in~\cite{willems2012open,baggio2017lti}, our pragmatic approach allows us to propose actionable, i.e., simple and tractable, identification and control techniques; see Section~\ref{sec:dd_control} later.

\subsubsection{Deterministic \gls{LTI} behaviors}
Deterministic restricted \gls{LTI} behaviors can be regarded as special cases of the linear Gaussian behavior in Definition~\ref{def:cov_model} with constant $B_t\equiv B$ and $\Sw_t \equiv 0$.
In this case, the trajectories coincide with the mean with probability one, and thus, $w_{[t,t+L-1]} = \mw \in \mathrm{im}(B)$ for all $t\geq0$.
Note that the support of $w_{[t,t+L-1]}$ is restricted to the subspace $\mathrm{im}(B)$ also in case $\Sw =B\Sigma^\mathrm{z}B^\top$ with some $\Sigma^\mathrm{z}\in\mathbb{S}_{\geq0}$. 
Then, the covariance is singular with $\mathrm{rank}(\Sw) = d$, and the trajectories can be written as $w_{[t,t+L-1]} = Bz_t$, where $z_t\sim\mathcal{N}(\mw,\Sigma^\mathrm{z})$.
For a general $\mw\notin\mathrm{im}(B)$, $w_{[t,t+L-1]}$ belongs to an affine space instead of a linear subspace, and hence, the Gaussian behavior would describe an affine system~\cite{padoan2023data}.

\subsubsection{Stochastic \gls{LTI} systems}
Now we show that the state-space representation of stochastic \gls{LTI} systems can also be recovered with a specific choice of $\mw$ and $\Sw_t$.
\begin{prop} \label{prop:SS_stochastic_equivalence}
Consider the stochastic \gls{LTI} system~\eqref{eq:stochastic-SS}.
Assume that $w_{[t,t+L-1]}=[u_{[t,t+L-1]}^\top~y_{[t,t+L-1]}^\top]^\top$ satisfies~\eqref{eq:stochastic-SS}, with $x_t \sim \mathcal{N}(\mu^\mathrm{x}_t,\Sigma^\mathrm{x}_t),~u_{[t,t+L-1]}\sim\mathcal{N}(\mu^\mathrm{u}_t,\Sigma^\mathrm{u}_t)$.
Assume further that $x_t$ and $u_{[t,t+L-1]}$ are independent.
Then $w_{[t,t+L-1]} \sim \mathcal{N}(\mw,\Sw_t)$ with
\small
\begin{align*}
        \mw & = \begin{bmatrix}
            0 & I_{mL} \\ \mathcal{O}_L & \mathcal{T}_L^\mathrm{u}
        \end{bmatrix} 
        \begin{bmatrix}
             \mu^\mathrm{x}_t \\ \mu^\mathrm{u}_t
        \end{bmatrix}, \\
        \Sw_t & = \begin{bmatrix}
            \Sigma^\mathrm{u}_t & \Sigma^\mathrm{u}_t (\mathcal{T}_L^\mathrm{u})^\top \\
            \mathcal{T}_L^\mathrm{u}\Sigma^\mathrm{u}_t & \mathcal{O}_L\Sigma^\mathrm{x}_t\mathcal{O}_L^\top + \mathcal{T}_L^\mathrm{u}\Sigma^\mathrm{u}_t(\mathcal{T}_L^\mathrm{u})^\top + \mathcal{T}_L^\xi\Sigma^\xi(\mathcal{T}_L^\xi)^\top + \Sigma^\eta
        \end{bmatrix}.
\end{align*}
\normalsize
\end{prop}
\vspace{2mm}
\begin{proof}
    Using~\eqref{eq:stochastic_ss_linear}, the trajectory can be expressed as a linear function of $[x_t^\top~ u_{[t,t+L-1]}^\top~\xi^\top~\eta^\top]^\top$, which is
    Gaussian with mean $[(\mu^\mathrm{x}_t)^\top~(\mu^\mathrm{u}_t)^\top~0^\top]^\top$ and block diagonal covariance matrix with blocks $\Sigma^\mathrm{x},\Sigma^\mathrm{u}_t,\Sigma^\xi$, and $\Sigma^\eta$.
\end{proof}

Proposition~\ref{prop:SS_stochastic_equivalence} shows that the mean $\mw$ and covariance $\Sw_t$ in Definition~\ref{def:cov_model} can be constructed from the stochastic state-space representation~\eqref{eq:stochastic-SS} and the covariances of $u,~x_t,~\xi$, and $\eta$.
Note that the initial state $x_t$ is assumed to be normally distributed, which is a common assumption in the covariance steering literature~\cite{chen2015optimal,liu2024optimal,pilipovsky2023data,aolaritei2022distributional} as well.
In contrast to the deterministic behaviors in the previous subsection, the covariance is always non-singular if $\Sigma^\mathrm{u}_t$ and $\Sigma^\eta$ are non-singular.
Note that the mean satisfies the linearity condition~\eqref{eq:linear} with $B$ being the matrix $\begin{bmatrix}
    0 & I_{mL} \\ \mathcal{O}_L & \mathcal{T}_L^\mathrm{u}
\end{bmatrix}$ with a suitable permutation of the rows, which is full rank in case the pair $(A,C)$ in~\eqref{eq:stochastic-SS} is observable.
Last, consider the \textit{stationary} case: namely, $\Sw_t$ is constant, if the system is at steady-state, i.e., $\Sigma^\mathrm{x}_t = \Sigma^\mathrm{x}$ and $\Sigma^\mathrm{u}_t = \Sigma^\mathrm{u}$ with $\Sigma^\mathrm{x} = A \Sigma^\mathrm{x} A^\top + B \Sigma^\mathrm{u}B^\top + \Sigma^\xi$.

In this section, we established that deterministic \gls{LTI} behaviors are consistent with our definition of Gaussian behaviors with zero, or constant, but singular covariance $\Sw$. 
Moreover, \gls{LTI} state-space systems in steady-state produce trajectories also with constant covariance matrix.
These observations indicate that Definition~\ref{def:cov_model} could be further refined, which we leave for future work. 
For the remainder of this paper, however, we consider stationary Gaussian behaviors with constant covariance $\Sw_t \equiv \Sw$, and do not take into account any structure in $\Sw_t$.
While apparently naive, this pragmatic perspective will serve us well in the data-driven setting.
Analogously, the naive ``a linear system is a subspace" without further refinement (e.g. on rank) served well in the deterministic setting~\cite{dorfler2022bridging,markovsky2022data,markovsky2021behavioral,huang2021quadratic}.

\subsection{Identifying linear Gaussian behaviors} \label{sec:sysID}
Recall the definition of the data matrix $W$ from Section~\ref{sec:BST}.
We propose to use the sample covariance of the data $\Swh = (1/D)\cdot WW^\top$ to estimate the covariance matrix $\Sw$.
It is worthwhile noting that in the deterministic \gls{LTI} setting, the sample covariance serves as a data representation of the behavior, since $\mathrm{im}(\Swh) = \mathrm{im}(W) \subseteq \mathcal{B}_L$.
Consequently, existing data-driven methods that exploit the image of $W$ can be directly interpreted in the proposed framework using the image of $\Swh$.
Furthermore, we show that $\Swh$ is the maximum likelihood estimate in the following special case.

\begin{lemma} \label{lemma:MLE}
    Assume that $w^i,~i\in\{1,\dots,D\}$ are i.i.d. samples from $\mathcal{N}(0,\Sw)$ and $W$ is full row rank. Then
    \begin{align*}
        \Swh =\frac{1}{D} \Wd\Wd^\top=\arg \max_{\Sw\in\mathbb{S}_{>0}} \Pi_{i=1}^{D}~p_w(w^i),
    \end{align*}
    where $p_w(w^i)$ denotes the probability density function of $\mathcal{N}(0,\Sw)$ evaluated at $w^i$.
\end{lemma}
\begin{proof}
The optimizer $\Swh$ can be expressed by taking the logarithm of the cost and using the first order optimality condition.
See, e.g.,~\cite[Thm. 3.1.5]{muirhead1982aspects} for the derivation.
\end{proof}

The rank condition on $\Wd$ can be interpreted as a persistence of excitation condition analogously to the \textit{generalized persistency of excitation} condition for deterministic \gls{LTI} behaviors~\cite{markovsky2022data}, which is equivalent to identifiability.
Furthermore, the trajectories are zero mean, if the system starts from zero mean Gaussian initial conditions and zero mean Gaussian input is applied to gather data. Then, due to linearity, $\mw \equiv 0$ holds.
In case the samples $w^i$ are linearly correlated, the maximum likelihood estimate is given as a weighted sample covariance~\cite{burda2011applying,zhang2019covariance}.

Now we turn to the estimation of the predictive distribution that is incorporated in the control formulations later.
For linear Gaussian behaviors, the mean satisfies $\mdep = M\mfree$ for some matrix $M$ due to~\eqref{eq:linear}.
Let us partition the data matrix as $W=[(W^\mathrm{free})^\top~(W^\mathrm{dep})^\top]^\top$ with blocks corresponding to the free and dependent variables.
Then, in a certainty-equivalent setting, the matrix $M$ can be estimated by the subspace predictor $\hat{M} = W^\mathrm{dep}(W^\mathrm{free})^\dagger = \hat{\Sigma}^\mathrm{df}_w (\hat{\Sigma}^\mathrm{ff}_w)^\dagger$, see~\cite{favoreel1999spc,dorfler2022bridging}, and $\mp,\Sp$ can be estimated as follows.

\begin{lemma} \label{lemma:pred_estim}
    For $\Swh=(1/D)\cdot W W^\top$ as in Lemma~\ref{lemma:MLE} and $\hat{M}=W^\mathrm{dep}(W^\mathrm{free})^\dagger$, the estimated predictive distribution is $\mathcal{N}(\mph,\Sph)$, with
\begin{align}
    \mph & = W^\mathrm{dep}(W^\mathrm{free})^\dagger
    \mathrm{w}^\mathrm{free}_t, \label{eq:pred_mean}\\ 
    \Sph & = \frac{1}{D} W^\mathrm{dep}\left(I-W^\mathrm{free}(W^\mathrm{free})^\dagger\right)(W^\mathrm{dep})^\top, \label{eq:var_mean}
\end{align}
\end{lemma}
\begin{proof}
    The formulas follow from~\eqref{eq:pred_mean_theory} and \eqref{eq:pred_var_theory} with $\Sigma^\mathrm{ff} = \frac{1}{D}W^\mathrm{free}(W^\mathrm{free})^\top$, 
    $\Sigma^\mathrm{fd} = \frac{1}{D}W^\mathrm{free}(W^\mathrm{dep})^\top$, and $~\Sigma^\mathrm{dd} = \frac{1}{D}W^\mathrm{dep}(W^\mathrm{dep})^\top$.
\end{proof}
Note that identifying the predictive distribution in the control setting was also proposed in~\cite{fiedler2023probabilistic}, leading to similar formulae.
In case $\Sw$ and $M$ are estimated as above, the terms containing $\mfree$ and $\mdep$ cancel out in~\eqref{eq:pred_mean_theory}, and thus, $\mph$ does not depend on $\mw$.
Consequently, the estimates in Lemma~\ref{lemma:pred_estim} can be readily applied in predictive control formulations.
However, $M\neq \Sigma^\mathrm{df} (\Sigma^\mathrm{ff})^\dagger$ in general, and $\mp$ in~\eqref{eq:pred_mean_theory} does depend on the unknown mean $\mw$.
To mitigate this error, we propose control formulations that are robust to deviations in $\mph$.

\section{Data-driven control with Gaussian behaviors} \label{sec:dd_control}
We propose stochastic predictive control formulations relying on our definition of Gaussian behaviors, and show that they are closely related to existing control formulations described in Section~\ref{sec:prelim-dd-control}.
We assume that a realization of the future input can be chosen, i.e., $\ufr$ is a deterministic decision variable.
The future outputs are random variables that depend on $\ufr$ and $\wir$, and they follow the predictive distribution from Section~\ref{sec:pred}.
Therefore, we define the control cost as 
$$
\tilde{J}(\ufr) := J([\ufr^\top~(\yf|\uf=\ufr,\wi=\wir)^\top]^\top),
$$
which is a family of random variables parametrized by $\ufr, \wir$.
The proposed control problems aim at finding $\ufr$ that minimizes the expected control cost given $\wir$.
Since $\yf$ is random, it is natural to include chance constraints on it in the following formulations, yet we omit them for simplicity.

\begin{rmk}
    Aside from predictive control formulations, our definition of Gaussian behaviors fits naturally with covariance steering problems~\cite{chen2015optimal,liu2024optimal,pilipovsky2023data,aolaritei2022distributional}.
    In covariance steering, the feedback controller that achieves a desired distribution for the closed-loop output usually consists of a linear state feedback gain and a feedforward term.
    In our setup, $\wi$ can be considered as a (non-minimal) initial state.    
    By tuning the feedback gain, the cross-correlation between $\uf$ and $\wi$ in the covariance $\Sw$ can be shaped.
    Therefore, $\Sp$ can be influenced by the controller design.
    Furthermore, the mean $\mp$ can be assigned directly by a feedforward term.   \hfill\QED
\end{rmk}

\subsection{Certainty-equivalence control}
We propose a control formulation that calculates the input by minimizing the expected control cost over the estimated predictive distribution from Section~\ref{sec:sysID}.
This is a certainty-equivalence approach, as we do not take into account the inaccuracy of the estimated model.
Let $\mph$ and $\Sph$ be defined as in~\eqref{eq:pred_mean} and~\eqref{eq:var_mean}.
Then, the control problem is
\begin{align} \label{eq:expected_cost}
\begin{split}
    \min_{\ufr\in\mathcal{U}} & \quad \mathbb{E}_{\mathcal{N}\left(\mph,\Sph \right)}[\tilde{J}(\ufr)].
\end{split}
\end{align}
The expected cost can be decomposed as $\|\ufr-\ur\|_{R}^2 + \|\mph-\yr\|_{R}^2 + \mathrm{tr}(Q\Sph)$. Since the predictive mean $\mph$ is estimated through the subspace predictor in~\eqref{eq:pred_mean}, it is equal to the deterministic $y_\mathrm{f}$ in the \gls{SPC} formulation~\eqref{eq:SPC}.
Moreover, the term $\mathrm{tr}(Q\Sph)$ does not depend on the decision variable $\ufr$.
Therefore, problem~\eqref{eq:expected_cost} is equivalent to problem~\eqref{eq:SPC} with $\mathcal{Y}=\mathbb{R}^{\dy}$.

\subsection{Distributionally optimistic control}
We now show that, contrary to the certainty-equivalence formulation in \eqref{eq:expected_cost}, the regularized \gls{DeePC} scheme~\eqref{eq:DeePC} is a distributionally optimistic formulation.
The predictive distribution in Lemma~\ref{lemma:pred_estim} may be inaccurate, and the mean in~\eqref{eq:pred_mean_theory} depends on the unknown mean $\mw$ in general.
To account for this uncertainty, we propose to 
optimize the predictive mean such that the resulting distribution remains close to the estimate in Lemma~\ref{lemma:pred_estim}, yet it leads to low expected control cost. 
We quantify the closeness of the two distributions by assuming an upper bound $\epsilon$ on their \gls{KL} divergence, leading to the distributionally optimistic problem
\begin{align} \label{eq:min-min-constrained}
\begin{split}
    \min_{\ufr \in \mathcal{U}} & \min_{\mu\in\mathbb{R}^{\dy}}  \quad  \mathbb{E}_{\mathcal{N}(\mu,\Sp)}[\tilde{J}(\ufr)] \\
    \mathrm{s.t.} & ~ D_{KL}\left(\mathcal{N}(\mu,\Sp)~\big\|~\mathcal{N}\left(\mph,
    \Sph\right) \right) \leq \epsilon.
\end{split}
\end{align}
The \gls{KL} divergence $D_{KL}$ measures the relative entropy between two distributions, and for Gaussians it is~\cite{williams2006gaussian}
\begin{align*}
     &D_{KL}\left(\mathcal{N}(\mu_1,\Sigma_1)~\big\|~\mathcal{N}(\mu_2,\Sigma_2)\right) = \frac{1}{2}\|\mu_1-\mu_2\|_{\Sigma_2^{-1}}^2 + c,
\end{align*}
where $c$ quantifies the mismatch between $\Sigma_1$ and $\Sigma_2$, and is independent of $\mu_1,\mu_2$.
If $\Sp$ is unknown, the constant $c$ in~\eqref{eq:min-min-constrained} cannot be computed.
However, as often done in robust formulations, we consider $\epsilon$ to be a tuning parameter, and thus, exact knowledge of $c$ is not required.
We use the \gls{KL} divergence in~\eqref{eq:min-min-constrained} to exploit this benefit and to highlight the similarities between the optimistic formulation and regularized \gls{DeePC}.
To simplify~\eqref{eq:min-min-constrained}, we lift the constraint on $D_{KL}$ to the cost with some $\lambda\geq0$ resulting in 
\begin{align} \label{eq:min_min}
\begin{split}
    & \min_{\ufr\in\mathcal{U},~\mu\in \mathbb{R}^{\dy}} \quad  \mathbb{E}_{\mathcal{N}(\mu,\Sp)}[\tilde{J}(\ufr)] \\
    & \quad + \lambda \cdot D_{KL}\left(\mathcal{N}(\mu,\Sp)~\big\|~\mathcal{N}\left(\mph,\Sph\right) \right).
\end{split}
\end{align}
Intuitively, larger values of the tuning parameter $\lambda$ correspond to smaller uncertainty $\epsilon$ in~\eqref{eq:min-min-constrained}, and vice versa.
The solution $\ufr^\star,\mu^\star$ of the two problems coincide, if $\lambda$ in \eqref{eq:min_min} takes the value of the optimal Lagrange multiplier of the constraint in \eqref{eq:min-min-constrained}.
Note that as $\epsilon\to0$, problem~\eqref{eq:min-min-constrained} collapses to the certainty-equivalence formulation~\eqref{eq:expected_cost}, and hence, we expect the same to hold for \eqref{eq:min_min} as $\lambda \to \infty$.

Next, we show that the projected 2-norm regularizer in the \gls{DeePC} method and the softened constraint on the \gls{KL} divergence in \eqref{eq:min_min} play the same role, and therefore the two problems have the same minimizer.
\begin{thm} \label{thm:DeePC_eqv}
If $(\wf^\star,g^\star)$ is a minimizer of \eqref{eq:DeePC} with $h(g) = \|(I-\Pi)g\|_2^2$ and $\mathcal{Y}=\mathbb{R}^{\dy}$, then $\ufr^\star = U_\mathrm{f} g^\star$ and $\mu^\star = \yf^\star = Y_\mathrm{f} g^\star$ is a minimizer of \eqref{eq:min_min} with $\lambda = \lambda_g \frac{2}{D}$.
\end{thm}
\if\arxiv0
The proof can be found in the extended version of the paper~\cite{sasfi2025gaussian}.
\fi
Note that solving~\eqref{eq:DeePC} is computationally more expensive than solving~\eqref{eq:min_min}, as the size of $g$ in~\eqref{eq:DeePC} scales with the number of data points $D$. 
Problem~\eqref{eq:min_min} accounts for the uncertainty in the estimated distribution in an optimistic fashion, as the mean $\mu$ is chosen such that the control cost is minimized.
This \textit{optimism in the face of uncertainty} has been observed to perform well by balancing exploitation and exploration in adaptive settings~\cite{abbasi2011regret,abeille2020efficient}. 
However, in the considered setup, one might expect the realized (closed-loop) control performance to degrade as the model uncertainty captured by $\epsilon$ increases, or equivalently, $\lambda$ decreases.
This intuition is supported by the empirical observation made in~\cite{dorfler2022bridging}, stating that the realized control performance monotonically increases as $\lambda_g$ in \eqref{eq:DeePC} is increased.

\if\arxiv1
\begin{proof}
    To prove the claim, we show that the costs of \eqref{eq:DeePC} and \eqref{eq:min_min} are equal up to a constant, and the feasible sets coincide. 
    Any term that does not depend on $\uf,\yf,\mu$ or $g$ will be denoted by $c$.
    First note that the expected control cost in \eqref{eq:min_min} can be written as
    $\mathbb{E}_{\mathcal{N}(\mu,\Sigma)}[\tilde{J}(\ufr)] = J([\ufr^\top~\mu^\top]^\top) + c$.
    
    Next, we show that $h(g)$ and the \gls{KL} divergence in \eqref{eq:min_min} are related.
    Let us decompose $g$ in \eqref{eq:DeePC} as
    \begin{align*}
        g = g_\mathrm{part} + g_\mathrm{hom},\quad g_\mathrm{part} \in \mathrm{im}(\Wd^\top),~ g_\mathrm{hom}\in \mathrm{ker}(\Wd),
    \end{align*} 
    where $\mathrm{ker}(\Wd)$ denotes the kernel of $\Wd$.
    Since $(I-\Pi)^\top(I-\Pi) = I-\Pi$, $g_\mathrm{part} \perp g_\mathrm{hom}$, and $(I-\Pi)g_\mathrm{hom} = 0$, the regularizer becomes
    \begin{align*}
        h(g) = \|g_\mathrm{part}\|_{(I-\Pi)}^2 + \|g_\mathrm{hom}\|_2^2.
    \end{align*}
    Note that $g_\mathrm{hom}\in\mathrm{ker}(W)$ can be removed from the constraint and only appears in the regularization term $h(g)$ in the optimization, therefore $g_\mathrm{hom} = 0$ for any optimal solution.
    Consequently, $g = g_\mathrm{part} = \Wd^\dagger \begin{bmatrix}
        \wi \\ \wf
    \end{bmatrix}$. 
    Furthermore, since $\begin{bmatrix}
        W_\mathrm{p} \\ U_\mathrm{f} 
    \end{bmatrix}g_\mathrm{part} = \begin{bmatrix}
        \wi \\\uf
    \end{bmatrix}$, the regularization on $g_\mathrm{part}$ can be expressed as
    \small
    \begin{align*}
        h(g_\mathrm{part}) = \left\|\begin{bmatrix}
            \wi \\ \wf
        \end{bmatrix}\right\|_{(\Wd\Wd^\top)^{-1}}^2 - \left\| \begin{bmatrix}
        \wi \\\uf
    \end{bmatrix} \right\|_{\left(\begin{bmatrix}
        W_\mathrm{p} \\ U_\mathrm{f} 
    \end{bmatrix}\begin{bmatrix}
        W_\mathrm{p} \\ U_\mathrm{f} 
    \end{bmatrix}^\top\right)^{-1}}^2.
    \end{align*}
    \normalsize
    Observe that for any $x$ and $M\in\mathbb{S}_{>0}$, the weighted norm satisfies $\|x\|_{M}^2 = -2\log(p(x)) - \tilde{c}$, where $p$ denotes the probability density function of $\mathcal{N}(0,M^{-1})$, and $\tilde{c}$ is a constant independent of $x$.
    Then
    \begin{align*}
         h(g_\mathrm{part}) = & -\frac{2}{D}\log\left(p_{w}([\wi^\top~\wf^\top]^\top)\right) \\
    & + \frac{2}{D}\log\left(p_\mathrm{m}([\wi^\top~\uf^\top]^\top)\right)+c,
    \end{align*}
    where $p_w$ is the density of $\mathcal{N}(0,\Sw)$, and $p_\mathrm{m}$ is density of the estimated marginal distribution over $\wi$ and $\uf$ given as 
    $\mathcal{N}\left(0,\begin{bmatrix}
       W_\mathrm{p} \\ U_\mathrm{f} 
    \end{bmatrix}\begin{bmatrix}
       W_\mathrm{p} \\ U_\mathrm{f} 
    \end{bmatrix}^\top\right)$.
    We can then apply the law of conditional probability $p_{\yf|\wi,\uf}(\yf|\wi,\uf) = p_w(w)/p_\mathrm{m}([\wi^\top~\uf^\top]^\top)$, where $p_{\yf|\wi,\uf}$ is the density of the predictive distribution $\mathcal{N}(\mph,\Sph)$, leading to
    \begin{align*}
        h(g_\mathrm{part}) = & -\frac{2}{D}\log(p_{\yf|\wi,\uf}(\yf))  + c \\
        = &\frac{1}{D} (\yf - \mph)^\top \Sph^{-1}(\yf - \mph) + c \\
        = & \frac{2}{D}D_{KL}\left(\mathcal{N}(\yf,\Sp)~\big\|~\mathcal{N}(\mph,\Sph)\right) + c.
    \end{align*}
    Thus, for any optimal value of $g$, the regularizer $h(g)$ is proportional to the \gls{KL} divergence plus a constant term.
\end{proof}
\fi

\subsection{Distributionally robust control}
While optimism in the face of uncertainty often works well, sometimes one needs to be conservative to minimize risks.
To this end, we propose a novel distributionally robust formulation that minimizes the expected control cost under the worst-case deviation from the estimated predictive distribution.
The distributionally robust problem is posed as
\begin{align} \label{eq:min-max-constrained}
    \begin{split}
    & \min_{\ufr \in \mathcal{U}} \max_{\mu\in\mathbb{R}^{\dy}} \;  \mathbb{E}_{\mathcal{N}(\mu,\Sp)}[\tilde{J}(\ufr)] \\
    & \quad  \mathrm{s.t.} \; D_{KL}\left(\mathcal{N}(\mu,\Sp)~\big\|~\mathcal{N}\left(\mph,\Sph\right) \right) \leq \epsilon.
    \end{split}
\end{align}

The cost of problem~\eqref{eq:min-max-constrained} is upper bounded by the dual of the inner maximization problem as formalized below.
\begin{thm} \label{thm:final}
    The cost of~\eqref{eq:min-max-constrained} is upper bounded by
    \begin{align} \label{eq:final_opt}
        & \mathbb{E}_{\mathcal{N}(\mu^\star,\Sp)}[\tilde{J}(\uft)] \\
    & -\lambda\left( D_{KL}\left(\mathcal{N}(\mu^\star,\Sp)~\big\|~\mathcal{N}(\mph,\Sph)\right) - \epsilon\right), \nonumber
    \end{align}
    where
    \begin{align*} 
        \mu^\star := (\lambda \Sph^{-1}-Q)^{-1}(\lambda\Sph^{-1}\mph -Q\yr),
    \end{align*}
    for any $\uft$ and any $\lambda\geq\lambda_0$, where $\lambda_0>0$ is such that $\lambda_0\Sph^{-1}-Q\in\mathbb{S}_{>0}$.
\end{thm}
\begin{proof}
    The inner problem in~\eqref{eq:min-max-constrained} can be expressed as $\min_{\mu}-\mathbb{E}_{\mathcal{N}(\mu,\Sp)}[\tilde{J}(\ufr)]$ subject to the constraint on $D_{KL}$, which is lower bounded by its dual function. Thus, the negative of the dual, given in~\eqref{eq:final_opt}, is an upper bound on the inner problem in~\eqref{eq:min-max-constrained}.
    Since~\eqref{eq:min-max-constrained} is minimized over $\ufr$, its cost is lower than~\eqref{eq:final_opt} for any $\uft$.
\end{proof}

As in problem~\eqref{eq:min_min}, $\lambda$ is the Lagrange multiplier of the constraint on $D_{KL}$ in the inner maximization problem in~\eqref{eq:min-max-constrained}, and we consider it to be a tuning parameter.
The mean $\mu^\star$ maximizes the Lagrangian of the inner problem in~\eqref{eq:min-max-constrained} by balancing the term $\|\mu-\yr\|_Q^2$ from the control cost and the term $\lambda\|\mu-\hat{\mu}\|_{\Sph^{-1}}^2$ from the constraint.

We choose the input $\ufr$ to minimize the upper bound~\eqref{eq:final_opt}, which can be rearranged to yield the control problem
\begin{align} \label{eq:final_cvx}
\begin{split}
    \min_{\ufr\in\mathcal{U}} & \quad \|\lambda\Sph^{-1}\mph - Q\yr\|_{(\lambda\Sph^{-1}-Q)^{-1}}^2 \\
    & \quad -\lambda\|\mph\|_{\Sph^{-1}}^2 + \|\ufr - \ur\|_{R}^2 + c,
\end{split}
\end{align}
where $\mph$ is a function of $\ufr$ given in~\eqref{eq:pred_mean} and $c$ denotes the terms independent of $\ufr$.
Since $\lambda$ is a Lagrange multiplier, increasing its value is directly related to considering a smaller uncertainty set, i.e., smaller $\epsilon$.
As $\lambda \to \infty$, we have $\mu^\star \to \mph$, and thus, we recover the certainty-equivalence formulation~\eqref{eq:expected_cost}.

Notice that the mean $\mph$ from~\eqref{eq:pred_mean} is affine in $\ufr$, and thus, the cost of~\eqref{eq:final_cvx} is quadratic in $\ufr$.
Let us write $\mph = \hat{M}_u \ufr + \hat{M}_\mathrm{ini}\wir$.
Then, the Hessian of the cost in~\eqref{eq:final_cvx} is
\begin{align*}
\begin{split}
    H(\lambda) := & \lambda^2\hat{M}_u^\top\Sph^{-1}\left(\lambda\Sph^{-1}-Q\right)^{-1}\Sph^{-1}\hat{M}_u \\
    & - \lambda\hat{M}_u^\top\Sph^{-1}\hat{M}_u + R.
\end{split}
\end{align*}
If $\lambda\geq\lambda_0$ is chosen such that $H(\lambda)$ is positive semidefinite, problem~\eqref{eq:final_cvx} is convex, and hence, can be solved efficiently.
Furthermore, since $H(\lambda) \to R$ as $\lambda\to\infty$, the convexity of the problem is guaranteed for large enough $\lambda$.
Preliminary numerical simulations showed that as $\lambda$ is increased, the optimal solution of~\eqref{eq:final_cvx} converges to that of~\eqref{eq:expected_cost}, as expected.
Furthermore, the distributionally optimistic and robust formulations in~\eqref{eq:min_min} and~\eqref{eq:final_cvx} yield similar realized control costs for the \gls{LTI} system considered in~\cite{dorfler2022bridging}.

\begin{rmk}
    Note that the dual of the inner maximization problem in~\eqref{eq:min-max-constrained} can also be expressed as a semidefinite programming, where $\lambda$ is an optimization variable~\cite[App. B.1]{boyd2004convex}.
    Then, strong duality holds if the constraint on the \gls{KL} divergence satisfies Slater’s condition, and thus, problem~\eqref{eq:min-max-constrained} can be reformulated as a convex program.
    However, since we treat $\lambda$ as a tuning parameter in the considered setting, it is simpler to only provide an upper bound in Theorem~\ref{thm:final}.
    \hfill \QED
\end{rmk}

\section{Conclusion}
We proposed a modeling framework for stochastic systems termed Gaussian behaviors, which is consistent with existing \gls{LTI} system models.
This description leads to a new perspective on data-driven predictive control methods from the literature, by interpreting regularization as accounting for errors in the estimated distribution in an optimistic fashion.
Furthermore, we proposed a novel distributionally robust control formulation that leads to a convex optimization problem.
Future work includes extending the results beyond Gaussians allowing us to model a larger class of systems.

\bibliographystyle{ieeetr}        
\bibliography{Literature}

\begin{thebibliography}{10}

\bibitem{coulson2019data}
J.~Coulson, J.~Lygeros, and F.~D{\"o}rfler, ``{Data-enabled predictive control: In the shallows of the DeePC},'' in {\em Proc. 18th European Control Conference}, pp.~307--312, IEEE, 2019.

\bibitem{coulson2021distributionally}
J.~Coulson, J.~Lygeros, and F.~D{\"o}rfler, ``Distributionally robust chance constrained data-enabled predictive control,'' {\em IEEE Trans. Automatic Control}, vol.~67, no.~7, pp.~3289--3304, 2021.

\bibitem{berberich2020data}
J.~Berberich, J.~K{\"o}hler, M.~A. M{\"u}ller, and F.~Allg{\"o}wer, ``Data-driven model predictive control with stability and robustness guarantees,'' {\em IEEE Trans. Automatic Control}, vol.~66, no.~4, pp.~1702--1717, 2020.

\bibitem{breschi2023data}
V.~Breschi, A.~Chiuso, and S.~Formentin, ``Data-driven predictive control in a stochastic setting: A unified framework,'' {\em Automatica}, vol.~152, p.~110961, 2023.

\bibitem{depersis2019formulas}
C.~De~Persis and P.~Tesi, ``Formulas for data-driven control: Stabilization, optimality, and robustness,'' {\em IEEE Trans. Automatic Control}, vol.~65, no.~3, pp.~909--924, 2019.

\bibitem{willems1997introduction}
J.~C. Willems and J.~W. Polderman, {\em {Introduction to mathematical systems theory: A behavioral approach}}.
\newblock Springer Media, 1997.

\bibitem{willems2005note}
J.~C. Willems, P.~Rapisarda, I.~Markovsky, and B.~L. De~Moor, ``A note on persistency of excitation,'' {\em Systems \& Control Letters}, vol.~54, no.~4, pp.~325--329, 2005.

\bibitem{dorfler2022bridging}
F.~D{\"o}rfler, J.~Coulson, and I.~Markovsky, ``Bridging direct and indirect data-driven control formulations via regularizations and relaxations,'' {\em IEEE Trans. Automatic Control}, vol.~68, no.~2, pp.~883--897, 2022.

\bibitem{markovsky2022data}
I.~Markovsky, L.~Huang, and F.~Dörfler, ``Data-driven control based on the behavioral approach: From theory to applications in power systems,'' {\em IEEE Control Systems Magazine}, vol.~43, pp.~28--68, 2023.

\bibitem{markovsky2021behavioral}
I.~Markovsky and F.~D{\"o}rfler, ``Behavioral systems theory in data-driven analysis, signal processing, and control,'' {\em Annual Reviews in Control}, vol.~52, pp.~42--64, 2021.

\bibitem{huang2021quadratic}
L.~Huang, J.~Zhen, J.~Lygeros, and F.~D{\"o}rfler, ``Quadratic regularization of data-enabled predictive control: Theory and application to power converter experiments,'' {\em IFAC-PapersOnLine}, vol.~54, no.~7, pp.~192--197, 2021.

\bibitem{willems2012open}
J.~C. Willems, ``Open stochastic systems,'' {\em IEEE Trans. Automatic Control}, vol.~58, no.~2, pp.~406--421, 2012.

\bibitem{baggio2017lti}
G.~Baggio and R.~Sepulchre, ``{LTI} stochastic processes: A behavioral perspective,'' in {\em Proc. 20th IFAC World Congress}, Elsevier, 2017.

\bibitem{faulwasser2023behavioral}
T.~Faulwasser, R.~Ou, G.~Pan, P.~Schmitz, and K.~Worthmann, ``{Behavioral theory for stochastic systems? A data-driven journey from Willems to Wiener and back again},'' {\em Annual Reviews in Control}, vol.~55, pp.~92--117, 2023.

\bibitem{pan2022stochastic}
G.~Pan, R.~Ou, and T.~Faulwasser, ``On a stochastic fundamental lemma and its use for data-driven optimal control,'' {\em IEEE Trans. Automatic Control}, vol.~68, no.~10, pp.~5922--5937, 2022.

\bibitem{fiedler2023probabilistic}
F.~Fiedler and S.~Lucia, ``Probabilistic multi-step identification with implicit state estimation for stochastic mpc,'' {\em IEEE Access}, vol.~11, pp.~117018--117029, 2023.

\bibitem{chiuso2025harnessing}
A.~Chiuso, M.~Fabris, V.~Breschi, and S.~Formentin, ``Harnessing uncertainty for a separation principle in direct data-driven predictive control,'' {\em Automatica}, vol.~173, p.~112070, 2025.

\bibitem{favoreel1999spc}
W.~Favoreel, B.~De~Moor, and M.~Gevers, ``{SPC: Subspace predictive control},'' {\em IFAC Proc. Volumes}, vol.~32, no.~2, pp.~4004--4009, 1999.

\bibitem{van2020willems}
H.~J. Van~Waarde, C.~De~Persis, M.~K. Camlibel, and P.~Tesi, ``Willems’ fundamental lemma for state-space systems and its extension to multiple datasets,'' {\em IEEE Contr. Syst. Lett.}, vol.~4, no.~3, pp.~602--607, 2020.

\bibitem{doob1953stochastic}
J.~L. Doob, {\em Stochastic Processes}.
\newblock John Wiley, New York, 1953.

\bibitem{williams2006gaussian}
C.~K. Williams and C.~E. Rasmussen, {\em Gaussian processes for machine learning}, vol.~2.
\newblock MIT press Cambridge, MA, 2006.

\bibitem{zhao2024data}
F.~Zhao, F.~Dörfler, A.~Chiuso, and K.~You, ``Data-enabled policy optimization for direct adaptive learning of the {LQR},'' {\em arXiv preprint arXiv:2401.14871}, 2024.

\bibitem{padoan2023data}
A.~Padoan, F.~D{\"o}rfler, and J.~Lygeros, ``Data-driven representations of conical, convex, and affine behaviors,'' in {\em Proc. 62nd Conference on Decision and Control}, pp.~596--601, IEEE, 2023.

\bibitem{chen2015optimal}
Y.~Chen, T.~T. Georgiou, and M.~Pavon, ``Optimal steering of a linear stochastic system to a final probability distribution, part i,'' {\em IEEE Trans. Automatic Control}, vol.~61, no.~5, pp.~1158--1169, 2015.

\bibitem{liu2024optimal}
F.~Liu, G.~Rapakoulias, and P.~Tsiotras, ``Optimal covariance steering for discrete-time linear stochastic systems,'' {\em IEEE Trans. Automatic Control}, vol.~70, no.~4, pp.~2289--2304, 2025.

\bibitem{pilipovsky2023data}
J.~Pilipovsky and P.~Tsiotras, ``Data-driven covariance steering control design,'' in {\em Proc. 62nd Conference on Decision and Control}, pp.~2610--2615, IEEE, 2023.

\bibitem{aolaritei2022distributional}
L.~Aolaritei, N.~Lanzetti, H.~Chen, and F.~D{\"o}rfler, ``Distributional uncertainty propagation via optimal transport,'' {\em arXiv preprint arXiv:2205.00343}, 2023.

\bibitem{muirhead1982aspects}
R.~J. Muirhead, {\em Aspects of Multivariate Statistical Theory}.
\newblock John Wiley \& Sons, Hoboken, NJ, 1982.

\bibitem{burda2011applying}
Z.~Burda, A.~Jarosz, M.~A. Nowak, J.~Jurkiewicz, G.~Papp, and I.~Zahed, ``Applying free random variables to random matrix analysis of financial data. part i: The gaussian case,'' {\em Quantitative Finance}, vol.~11, no.~7, pp.~1103--1124, 2011.

\bibitem{zhang2019covariance}
W.~Cui, X.~Zhang, and Y.~Liu, ``Covariance matrix estimation from linearly-correlated {Gaussian} samples,'' {\em IEEE Trans. Signal Processing}, vol.~67, no.~8, pp.~2187--2195, 2019.

\bibitem{abbasi2011regret}
Y.~Abbasi-Yadkori and C.~Szepesv\'ari, ``Regret bounds for the adaptive control of linear quadratic systems,'' in {\em Proc. 24th Annu. Conf. Learn. Theory}, pp.~1--26, 2011.

\bibitem{abeille2020efficient}
M.~Abeille and A.~Lazaric, ``Efficient optimistic exploration in linear-quadratic regulators via lagrangian relaxation,'' in {\em International Conference on Machine Learning}, pp.~23--31, 2020.

\bibitem{boyd2004convex}
S.~Boyd and L.~Vandenberghe, {\em Convex optimization}.
\newblock Cambridge university press, 2004.

\end{thebibliography}
\end{document}